\documentclass{X}

\RendicontiPagina{1}{Mattia Scomparin}{Mattia Scomparin}

\usepackage{amsmath}
\usepackage{amsfonts}
\usepackage{slashed}
\usepackage{xcolor} 

\newcommand{\p}{\phi}
\newcommand{\pc}{\phi^*}
\newcommand{\vp}{\varphi}
\newcommand{\vt}{\vartheta}
\newcommand{\ve}{\varepsilon}
\newcommand{\Lag}{\mathcal{L}(x,\p,\partial\p)}
\newcommand{\Lageps}{\mathcal{L}(x,\p_\ve,\partial\p_\ve)}

\title{Conserved currents from nonlocal constants in relativistic scalar field theories}
\author{Mattia Scomparin}
\dedicated{mattia.scompa@gmail.com}

\summary{
Nonlocal constants are functions that are constant 
along motion but whose value depends on the past history of the motion itself.
They are a powerful tool to provide first integrals in classical mechanics and, in this respect, a new approach to get nonlocal constants within the framework of lagrangian scalar field theory is introduced. We derive locally-conserved currents from them, and we prove the consistency of our 
results by recovering some standard Noetherian results. Applications include the real/complex 
nonlinear interacting theory and the real dissipative Klein-Gordon theory. 
}

\intesta

\begin{document}
\maketitle

\section{\label{sec:intro}Introduction}

A fundamental and common idea in modern physics is to describe the foundation of a theory via its Lagrangian. Within the framework of \textit{scalar field theories} embedded in a flat $D$-dimensional spacetime $\mathcal{M}^D$\!, a \textit{Lagrangian density} is a scalar, relativistically invariant, smooth function
\begin{equation}\label{eq:lag}
\Lag\,,
\end{equation}
with $\p(x)$ a scalar field that depends on spacetime events $x\in\Omega\subset\mathcal{M}^D\!$.  
By standard arguments in \textit{Calculus of Variations}, a fixed-endpoint, stationarizing motion for the action  $S\equiv\int d^D\!x\,\mathcal{L}$, called ``geodesic'' flow, is a solution $x\mapsto \p(x)$ to the \textit{Euler-Lagrange} field equation
\begin{equation}\label{eq:eom}
\frac{\delta\Lag}{\delta\p}-\partial_\mu\frac{\delta\Lag}{\delta\partial_\mu\p}=0\,.
\end{equation}
Here, the $\delta/\delta\p$ and $\delta/\delta\partial\p$ symbols stand for functional derivatives, whereas we use  $\partial_\mu\equiv\partial/\partial x^\mu$ for partial derivatives with respect to the spacetime coordinate $x^\mu$. 
Later on, Einstein's convention on dummy indices is assumed.
We also define $\eta_{\mu\nu}$ as the generalized, $D$-dimensional, Minkowski metric for $\mathcal{M}^D$.

A conserved current for theory \eqref{eq:lag} is a $D$-vector
\begin{equation}\label{eq:curr}
j^\mu(x,\p,\partial\p)\,,
\end{equation}
that is conserved along the solutions of field equation \eqref{eq:eom}, namely that satisfies the continuity equation 
\begin{equation}
\partial_\mu j^\mu(x,\p,\partial\p)=0\quad \textnormal{for all } x\in\Omega.
\end{equation}
In practical terms, finding conserved currents is one of the best approaches to understand the local/global behavior of a theory and provide insights to solve field equations.  
A beautiful and celebrated result in this respect is what is nowadays referred to as \textit{Noether's Theorem} \cite{Noether:1918,Rosen:1972ku}, which captures the strong connection between conserved currents and continuous symmetries of Lagrangian densities. 
Despite being a very powerful tool, Noether's Theorem assumes that theory's invariance must be satisfied at each point in field space, but finding symmetries is not always simple when we come to highly complicated Lagrangian densities.

Exploring systems with finite degrees of freedom, Gorni and Zampieri revisited Noether's Theorem in Ref. \cite{Zampieri:2013ziw} introducing the concept of \textit{nonlocal constants of motion}, that are functions that look like this:
\begin{equation}\label{eq:nonlocGZ}
\frac{\partial L(t,q(t),\dot{q}(t))}{\partial \dot{q}}\cdot \frac{\partial q_\ve(t)}{\partial \ve}\bigg\rvert_{\ve=0}-\int_{t_0}^t 
\frac{\partial L(s,q_\ve(s),\dot{q}_\ve(s))}{\partial\ve}\bigg\rvert_{\ve=0}ds\,,
\end{equation}
where $L(t, q, \dot{q})$ is a Lagrangian, $q,\dot{q}=dq/dt\in\mathbb{R}^n$ are generalized coordinates, and $q_\ve\in\mathbb{R}^n$ is an arbitrary perturbed motion depending on $\ve\in\mathbb{R}$. On technical grounds, expression \eqref{eq:nonlocGZ} is constant along solutions of equation \eqref{eq:eom} without necessarily requiring invariance proprieties of the Lagrangian. However, on the other hand, it is nonlocal since its value at time $t$ also depends on the whole history of the motion between $t_0$ and $t$. This result was successfully applied by Gorni and Zampieri in Refs. \cite{gorni2017,gorni2019} to study some relevant mechanical systems, which include: (i) homogeneous potentials of degree $k=-2$, (ii) mechanical systems with viscous fluid resistance, (iii) mechanical system with hydraulic (quadratic) fluid resistance, and (iv) conservative and dissipative Maxwell-Bloch equations of laser dynamics.
Beside that, Scomparin's generalization of nonlocal constants to higher-order Lagrangian systems $L(t,q,\dot{q},\ddot{q},...)$ appeared in Ref. \cite{scomparin2021}, with applications covering the Pais-Uhlenbeck oscillator and some models of modified gravity. In any case, standard Noetherian results were recovered.

Noether's work was originally motivated by \emph{General Relativity} \cite{Noether:1918}, and to this day some of the  most relevant applications of the theorem are still in field theories.
It is therefore questionable if the approach designed in Ref. \cite{Zampieri:2013ziw} could be extended to find (integral) nonlocal constants and conserved currents in such framework. In particular, since field theories like \eqref{eq:lag} appear in a very wide class of models for cosmological theory \cite{Arbey:2019cpf}, quantum theory \cite{Peskin:1995ev, Moshe:2003, Weinberg:1995mt}, thermodynamics \cite{Ballesteros:2016kdx}, and statistical physics \cite{TakashiHara}, an in-depth investigation is strongly motivated. 
On the other hand, approaching Noether's Theorem from a nonlocal point of view could provide new perspectives to find solutions. In fact, Euler-Lagrange field equations \eqref{eq:eom} often turn out to be non-linear Partial Differential Equations (PDE) hard to be solved. This goal was anticipated by the author in a recent review \cite{GORNI2022100262}.

It is important to mention that nonlocal constants and nonlocal symmetries in field theory entered the literature through some works of Peter Olver \cite{10.1007/978-94-007-1023-8_8}, Kossman-Schwarzbach \cite{dddddddd}, L. Dolan and A. Roos \cite{PhysRevD.22.2018}, and others.

The central purpose of this paper is to extend, in a systematic way,  the nonlocal approach developed by Gorni and Zampieri \cite{Zampieri:2013ziw} and Scomparin \cite{scomparin2021} to the framework of field theories. Since their nonlocal approach has not been implemented in such framework until now, we focus this groundwork study on canonical scalar field theories  \eqref{eq:lag} only, being their formulation sufficiently simple that all the relevant applications can be easily derived, recovered, and fully analyzed. Indeed, as mentioned, scalar fields appear in a wide range of well-functioning models, making the choice not too restrictive to get useful results. Generalizations to more sophisticated field theories are left to future investigation. 

The structure of the paper is organized as follows. The main, very simple result of our work is presented in Sec. \ref{sec:nonloconst}, where we deduce Theorem \ref{teo:nonloconst} to obtain nonlocal constants from Lagrangian densities like \eqref{eq:lag}. We find nonlocal constants that are the covariant transposition of expression \eqref{eq:nonlocGZ}, where a volume integral on $\Omega$ appears in the second term as a nonlocal contribute with respect to a surface integration in the first one. 
In this context, perturbed fields are treated as spacetime and field transformations.

Generally, nonlocal constants are trivial or of no apparent practical value \cite{gorni2017, scomparin2021}. However, there are few and precious Lagrangian densities for which  Theorem  \ref{teo:nonloconst} does yield  conserved currents, in the sense of local functions like \eqref{eq:curr} whose total divergence is zero. Indeed, inspired by previous studies \cite{Zampieri:2013ziw,gorni2017, scomparin2021}, our Theorem \ref{teo:NNN} in Sec. \ref{sec:conscurr} exhibits that an interesting interplay between nonlocal constants and conserved currents exists. More specifically, an unifying scheme emerges every time a Lagrangian density, once evaluated on a well thought perturbed field, satisfies what we call the ``total divergence condition''. Hence, driven by a symmetry-oriented approach, we provide a set of useful theorems giving conserved currents from nonlocal constants by using spacetime transformations (Subsec. \ref{sec:1}), internal transformations (Subsec. \ref{sec:2}), mixed spacetime-internal transformations (Subsec.s \ref{sec:3}, \ref{sec:4}, \ref{sec:6}), and finite invariances (Subsec. \ref{sec:5}) as perturbed fields. Some of them will be familiar from the literature \cite{Zampieri:2013ziw,gorni2017, scomparin2021}, others have been found, case by case, by inspection of the Lagrangian density taken into consideration. 
In particular, in Subsec. \ref{sec:1} Theorem \ref{teo:trasl} deals with spacetime-independent theories and recovers the well-known conservation of the \textit{Energy-Momentum tensor} $T^\mu_\nu$. In Subsec. \ref{sec:2}, Theorem \ref{teo:ssf}  gets conserved currents from Lagrangian densities such that $\delta \mathcal{L}/\delta\partial_\mu\p\propto \partial^\mu(\delta\mathcal{L}/\delta\p)$.  
In Subsec. \ref{sec:3} our Corollary \ref{teo:ssf3} considers theories whose conserved currents are generated by the Energy-Momentum tensor. In Subsec. \ref{sec:4}, Theorem \ref{teo:hom2} shows that two  conserved currents involving $T^\mu_\nu$ exist for Lagrangian densities of the form $\mathcal{L}=\tfrac{1}{2}\partial_\mu\p\,\partial^\mu\p-U(\p)$ if $U(\p)$ is a \textit{homogeneous} potential of degree $k=D\Delta_D^{-1}$. In this regard we also find equation \eqref{eq:DIST}, which exactly gives how the distance from the origin $|\p|$ depends on spacetime, even though we don't know the shape of the solution. In Subsec. \ref{sec:6},  Theorem \ref{teo:diss} yields a conserved current for dissipative theories of the form $\mathcal{L}=e^{\,h\, x}\big[\tfrac{1}{2}\partial_\mu\p\,\partial^\mu\p-U(\p)\big]$, with $U(\p)$ a homogeneous  potential of degree $k=2$. Finally, in Subsec. \ref{sec:5} our last Theorem \ref{teo:J} considers finite invariances of Lagrangian densities, i.e. a generalization of Noetherian symmetries, to generate conserved currents. 

Our theorems consistently generalize a set of well-known results obtained in Refs \cite{Logan:1977, Zampieri:2013ziw, gorni2017, scomparin2021} for mechanical systems with finite degrees of freedom. In particular, our theorems successfully recover such results by imposing $D=1$.

In the rest of the paper, Sec. \ref{sec:app}, we deal with some neat applications of
our theorems to relevant theories involving scalar fields. 
In Subsec. \ref{sub:KG} we consider the \textit{nonlinear interacting (real) theory} framework, which includes, among all, (i) the free massless theory \cite{Faber:2001dv}, (ii) the free Klein-Gordon theory \cite{Chiarelli:2017rdc}, and (iii) the interacting $\p^4$-theory \cite{Brandt:1975xn}, as special cases.  
First, we exhibit that a random choice of perturbed fields that satisfy the total divergence condition easily results conserved currents that are trivial. Then, we use the theorems designed in the previous section to recover a set of well-known Noetherian results from a nonlocal point of view. Interestingly, Theorem \ref{teo:hom2} applied to the massless $\p^4$-theory returns a very useful conserved current \eqref{eq:homHfdfd} that returns expression  \eqref{eq:DISTmass}, which gives the spacetime dependence of $|\p|$.

In Subsec. \ref{sub:CKG} we move to the \textit{nonlinear interacting (complex) theory} with $\p\in\mathbb{C}$.
Generally, complex scalar fields represent spin-0 particles and antiparticles with electric charge or other properties.
Again, treating $\p$ and its complex-conjugate $\pc$ as two independent fields, the full consistency of our machinery is confirmed by the recovered currents.
In particular, such currents generalize the results of Subsec. \ref{sub:KG}, and expression \eqref{eq:homHfdfdC} is the analogous of \eqref{eq:DISTmass}.

Subsec. \ref{sec:dissKG} is dedicated to apply our Theorem \ref{teo:diss} to the \textit{dissipative (real) Klein-Gordon theory}, where a multiplicative exponential damping factor modifies the usual free Klein-Gordon Lagrangian density. This example leads to the  conserved current \eqref{eq:dissjs}.
Finally, we notice that in Subsec. \ref{sub:KG} Theorem \ref{teo:trasl} and Theorem \ref{teo:ssf} result the same conserved quantity for the free Klein-Gordon theory. In Subsec. \ref{sub:spacetimeex}, we exhibit how to remove such degeneracy by considering a theory where the Lagrangian coefficients are not constant but have an explicit dependence on spacetime. In such case, between the two theorems, only Theorem \ref{teo:ssf} can be applied since the theory is not spacetime-independent.

As a final remark, we stress that nonlocal constants and corresponding conserved currents are a valuable tool also for studying models that, like the free  Klein-Gordon theory, easily exhibit a general solution without requiring a conservation law. Under this sense, as already analyzed in recent works (see e.g. Refs.  \cite{gorni2019,Kaparulin}), we believe our results could provide a valuable instrument to give a novel insight into stability proprieties of scalar field theories and boundedness of related solutions. We leave a complete analysis of these issues for future investigation, here we just focus on the formal implementation of our nonlocal approach.


\section{Nonlocal constants}\label{sec:nonloconst}

In this section, we introduce the notion of perturbed fields and outline how nonlocal constants can be obtained in scalar field theories.
For our purposes, the subset $\Omega\subset\mathcal{M}^D$ is supposed to be compact and with a smooth boundary $\partial\Omega$. We assume that $\partial\Omega$ is characterized by an outward pointing unit normal $n_\mu$ at each point.
\begin{definition}\label{defin:family}
Given a solution $x\mapsto \p(x)$ of the field equation \eqref{eq:eom}, a one parameter family of perturbed fields, or simply a perturbed field, associated to $\p(x)$ is a smooth function $(\ve,x)\mapsto \p_\ve(x)$ with $\ve\in \mathbb{R}$ in a neighbourhood of $0$, and such that $\p_0(x)=\p(x)$.
\end{definition}

By adopting the standard Noetherian terminology used in field theories, perturbed fields can be classified in \textit{spacetime families}, that act  on spacetime only, and \textit{internal families}, that transform the fields into each other in some way without making reference to their dependence on spacetime. Among all, elementary perturbed fields are the \textit{spacetime-shift family} $\p_\ve(x)=\p(x+\ve f(x))$, and the \textit{internal-shift family} $\p_\ve(x)=\p(x)+\ve g(x)$, with $f$ and $g$ free smooth functions of $x\in\Omega$.
%
\begin{theorem}\label{teo:nonloconst}
Let $x\mapsto\p(x)$ be a solution of the field equation \eqref{eq:eom} and let $\p_\ve(x)$ be a perturbed field associated to $\p(x)$. Then, the following function of $x\in\Omega$ is constant 
\begin{equation}\label{eq:nonloconst}
\oint_{\partial\Omega}\!d^{D-1}x\, n_\mu\frac{\delta \Lag}{\delta\,\partial_\mu\p}\,\frac{\partial\p_\ve}{\partial\ve}\bigg\rvert_{\ve=0}
-\int_{\Omega}\!d^D\!x\,\frac{\partial \Lageps}{\partial\ve}\bigg\rvert_{\ve=0}\,.
\end{equation}
\end{theorem} 
%
\begin{proof}
Take the derivative of $\Lageps$ with respect to $\ve$ at $\ve=0$, and use $\partial (\partial_\mu\p_\ve)/\partial\ve\rvert_{\ve=0}=\partial_\mu(\partial \p_\ve/\partial\ve)\rvert_{\ve=0}$ to get
\begin{equation}\label{eq:manip}
\begin{split}
\frac{\partial \Lageps}{\partial\ve}\bigg\rvert_{\ve=0}&=\frac{\delta \mathcal{L}}{\delta\p}\frac{\partial\p_\ve}{\partial\ve}\bigg\rvert_{\ve=0}+\frac{\delta \mathcal{L}}{\delta\partial_\mu\p}\partial_\mu\frac{\partial\p_\ve}{\partial\ve}\bigg\rvert_{\ve=0}\\
&=\partial_\mu \left[\frac{\delta \mathcal{L}}{\delta\partial_\mu\p}\frac{\partial\p_\ve}{\partial\ve}\bigg\rvert_{\ve=0}\right]+\left[\frac{\delta\mathcal{L}}{\delta\p}-\partial_\mu\frac{\delta\mathcal{L}}{\delta\partial_\mu\p}\right]\frac{\partial\p_\ve}{\partial\ve}\bigg\rvert_{\ve=0}\,.
\end{split}
\end{equation}
Using equation \eqref{eq:eom}, the second term on the right-hand-side of \eqref{eq:manip} disappears. The final result \eqref{eq:nonloconst} is obtained by integrating in the whole $\Omega$ domain and applying the covariant Gauss divergence theorem.
\end{proof}

Following the same nomenclature recommended by Gorni and Zampieri \cite{Zampieri:2013ziw} and by Scomparin \cite{scomparin2021}, expression \eqref{eq:nonloconst} will be called \textit{nonlocal constant} associated to $\p_\ve$. Indeed, a volume integral on $\Omega$ appears in the second term as a nonlocal contribute with respect to a surface integration in the first one. 
It should be noticed that expression \eqref{eq:nonloconst} is the covariant transposition of the result \eqref{eq:nonlocGZ} proposed by Gorni and Zampieri in Ref. \cite{Zampieri:2013ziw}, where the time is replaced by spacetime events, and $q(t)$ is replaced by the value of the field $\p(x)$ at that point in spacetime. In contrast, an integral structure appears in both the terms since the Lagrangian density is now embedded in a $D$-dimensional spacetime.

In Theorem \ref{teo:nonloconst} the one parameter family of perturbed fields can be chosen randomly, providing generally trivial or of no apparent practical interest nonlocal constants. 
This evidence will be clear when we will address some applications in Sec. \ref{sec:app}.
However, there exist well-thought cases where expression \eqref{eq:nonloconst} returns a local conserved current like \eqref{eq:curr}, as we will see in the next Sec. \ref{sec:conscurr}.


\section{\label{sec:conscurr}Conserved currents}

The following section describes particular cases for which Theorem \ref{teo:nonloconst} yields conserved currents for equation \eqref{eq:eom}.
In accordance with previous works (see e.g. Refs. \cite{Zampieri:2013ziw,scomparin2021}), this purpose justifies the following terminology:
\begin{definition}\label{def:totdiv}
Let $\p_\ve$ be a perturbed field associated to a solution $x\mapsto \p(x)$ of the field equation \eqref{eq:eom}. We say that $\p_\ve$ satisfies the ``total divergence condition'' with the $D$-vector $\psi^\mu(x,\p,\partial\p)$ if
\begin{equation}\label{eq:totdiv}
\frac{\partial \Lageps}{\partial\ve}\bigg\rvert_{\ve=0}=\partial_\mu\psi^\mu(x,\p,\partial\p)\quad \textnormal{for all } x\in\Omega.
\end{equation}
\end{definition}
Definition \ref{def:totdiv} refers to solution of equation \eqref{eq:eom} only. In principle one could define a total divergence condition considering a general $\p(x)$, whether it solves Euler-Lagrange equations \eqref{eq:eom} or not, but this is not our case.

The condition \eqref{eq:totdiv} becomes interesting when applied to Theorem \ref{teo:nonloconst}.

\begin{theorem}\label{teo:NNN}
Let $x\mapsto \p(x)$ be a solution of the field equation \eqref{eq:eom}, and let $\p_\ve$ be a perturbed field associated to $\p(x)$ satisfying the total divergence condition \eqref{eq:totdiv} for some $\psi^\mu(x,\p,\partial\p)$. Then, $j^\mu(x,\p,\partial\p)$ is a conserved current for equation \eqref{eq:eom}, with
\begin{equation}\label{eq:jtotdiv}
j^\mu(x,\p,\partial\p)=\frac{\delta \Lag}{\delta\,\partial_\mu\p}\,\frac{\partial\p_\ve}{\partial\ve}\bigg\rvert_{\ve=0}-\psi^\mu(x,\p,\partial\p)\quad \textnormal{for all } x\in\Omega.
\end{equation}
\end{theorem}
\begin{proof}
Add definition \eqref{eq:totdiv} inside expression \eqref{eq:nonloconst} and apply the covariant Gauss divergence theorem to the term containing $\partial_\mu\psi^\mu$. Writing the nonlocal constant as $\oint_{\partial\Omega}d^{D-1}\!x\, n_\mu j^\mu$, expression \eqref{eq:jtotdiv} holds for all $x\in\Omega$.
\end{proof}

At this point, we are able to get conserved currents without necessarily requiring a general invariance theory on the Lagrangian.
Indeed, after having sought a $\psi^\mu(x,\p,\partial\p)$ satisfying the total derivative condition \eqref{eq:totdiv}, we are naturally led by Theorem \ref{teo:NNN} to consider expression \eqref{eq:jtotdiv} as a conserved current. 
Generally we cannot expect to find such a $\psi^\mu(x,\p,\partial\p)$ for a random choice of perturbed field $\p_\ve$.
However there are few and precious Lagrangians that make the research easier.


\subsection{Spacetime-independent theories}
\label{sec:1}

As stated in many reference works on application of Noether's Theorem in scalar field theories (see e.g. Refs \cite{Callan1970,Forger:2003ut,Weinberg:1995mt}), spacetime-independent theories, i.e. whose Lagrangian densities $\mathcal{L}(\p,\partial\p)$ do not have an explicitly dependence on spacetime, are a classical and simple prototype to find conserved currents. 

\begin{theorem}\label{teo:trasl}
Let $x\mapsto \p(x)$ be a solution of the field equation \eqref{eq:eom} for a spacetime-independent Lagrangian density $\mathcal{L}(\p,\partial\p)$. Then, the Energy-Momentum tensor $T^\mu_\nu(\p,\partial\p)$ is a conserved current for equation \eqref{eq:eom}, namely $\partial_\mu T^\mu_\nu=0$, with
\begin{equation}\label{eq:T}
T^\mu_\nu\equiv\frac{\delta\mathcal{L}}{\delta\partial_\mu\p}\partial_\nu\p-\mathcal{L}\,\delta^\mu_\nu\,.
\end{equation}
\end{theorem}
\begin{proof}
Consider the spacetime-shift family $\p_\ve(x)=\p(x+\ve a)$ with the constant parameter $a^\mu\in\mathbb{R}^D$. It follows that $\partial\p_\ve/\partial\ve\rvert_{\ve=0}=a^\mu\partial_\mu\p$, so we can compute 
\begin{equation}
\begin{split}
\frac{\partial \Lageps}{\partial\ve}\bigg\rvert_{\ve=0}&=
\frac{\delta\mathcal{L}}{\delta\p}\frac{\partial\p_\ve}{\partial\ve}\bigg\rvert_{\ve=0}+\frac{\delta\mathcal{L}}{\delta\partial_\mu\p}\partial_\mu\frac{\partial\p_\ve}{\partial\ve}\bigg\rvert_{\ve=0}\\
&=\frac{\delta\mathcal{L}}{\delta\p}a^\mu\partial_\mu\p+\frac{\delta\mathcal{L}}{\delta\partial_\mu\p}a^\nu\partial_\mu\partial_\nu\p\\
&=\partial_{\mu}\big(\mathcal{L}\,\delta^\mu_\nu a^\nu\big)\,.
\end{split}
\end{equation}
Observe that $\p_\ve$  satisfies the total divergence condition \eqref{eq:totdiv} with $\psi^\mu=\mathcal{L}\,\delta^\mu_\nu a^\nu$. Hence, applying Theorem \ref{teo:NNN} we get
\begin{equation}
j^\mu_{(a)}=\left[\frac{\delta \mathcal{L}}{\delta\,\partial_\mu\p}\,\partial_\nu\p - \mathcal{L} \,\delta^\mu_\nu\right] a^\nu\quad \textnormal{for all } x\in\Omega.
\end{equation}
Since there are $D$ independent generators for each possible choice of translation parameters $a^\nu\!$, the final result \eqref{eq:T} is obtained by defining $j_{(a)}^\mu\equiv T^\mu_\nu\, a^\nu$.
\end{proof}
As expected, expression \eqref{eq:T} recovers from a nonlocal point of view the  well known Noetherian Energy-Momentum tensor (see e.g. Refs. \cite{Forger:2003ut,FREEDMAN1974354}). This fact should not surprise, since Theorem \ref{teo:trasl} assumes a spacetime-shift symmetric Lagrangian density. It is simple to prove that the Energy-Momentum tensor is symmetric, namely $T^{\mu\nu}=T^{\nu\mu}$. Here $T^{\mu\nu}=\eta^{\mu\sigma}T^\nu_\sigma$.


\subsection{\label{sec:2}Theories with $\delta\mathcal{L}/\delta\partial_\mu\p\propto \partial^\mu(\delta\mathcal{L}/\delta\p)$}

Within the framework of theories with finite degrees of freedom, Scomparin provided in Ref. \cite{scomparin2021} a new nonlocal approach to get the conservation of energy for both the canonical and the  higher-order harmonic oscillators, whose Lagrangians are quadratic in $q$ and derivatives.
In continuity with such result, we deduce that generally the total divergence condition \eqref{eq:totdiv} seems to be too easy to be satisfied starting from field-shift perturbed families if $\delta \mathcal{L}/\delta\partial_\mu\p\!\propto \partial^\mu\delta\mathcal{L}/\delta\p$.

\begin{theorem}\label{teo:ssf}
Consider a Lagrangian density $\Lag$ such that there exists a constant parameter $\rho\in \mathbb{R}$ such that for all motions, whether solutions or not of \eqref{eq:eom},
\begin{equation}\label{eq:ssfcondition}
\frac{\delta \mathcal{L}}{\delta\partial^\mu\p}=\rho\,\partial_\mu\frac{\delta\mathcal{L}}{\delta\p}\quad \textnormal{for all } x\in\Omega\,.
\end{equation}
Let $x\mapsto \p(x)$ be a solution of the field equation \eqref{eq:eom} for $\mathcal{L}$. Then, ${K}^{\mu}_\nu(x,\p,\partial\p)$ is a conserved current for equation \eqref{eq:eom}, namely $\partial_\mu{K}^{\mu}_\nu=0$, with
\begin{equation}\label{eq:D}
{K}^{\mu}_\nu\equiv
\frac{\delta\mathcal{L}}{\delta\partial_\mu\p}
\frac{\delta\mathcal{L}}{\delta\partial^\nu\p}
-\frac{1}{2}\delta^\mu_\nu\left[\rho\,\frac{\delta\mathcal{L}}{\delta\p}\frac{\delta\mathcal{L}}{\delta\p}+
\frac{\delta\mathcal{L}}{\delta\partial^\alpha\p}\frac{\delta\mathcal{L}}{\delta\partial_\alpha\p}\right].
\end{equation}
\end{theorem}
\begin{proof}
Consider the field-shift family $\p_\ve=\p+\ve \,b^\mu(\delta\mathcal{L}/\delta\partial^\mu\p)$ with  the constant parameter $b^\mu\in \mathbb{R}^D$. It follows that $\partial\p_\ve/\partial\ve\rvert_{\ve=0}=b^\mu(\delta\mathcal{L}/\delta\partial^\mu\p)$, so we can compute
\begin{equation}\label{eq:D1}
\begin{split}
\frac{\partial \Lageps}{\partial\ve}\bigg\rvert_{\ve=0}&=
\frac{\delta\mathcal{L}}{\delta\p}\frac{\partial\p_\ve}{\partial\ve}\bigg\rvert_{\ve=0}+\frac{\delta\mathcal{L}}{\delta\partial_\alpha\p}\partial_\alpha\frac{\partial\p_\ve}{\partial\ve}\bigg\rvert_{\ve=0}\\
&=\left[\rho\frac{\delta\mathcal{L}}{\delta\p}
\,\partial_\nu\frac{\delta\mathcal{L}}{\delta\p}+
\frac{\delta\mathcal{L}}{\delta\partial_\alpha\p}\,\partial_\alpha
\frac{\delta\mathcal{L}}{\delta\partial^\nu\p}\right]b^\nu\\
&=\partial_{\mu}\left[\frac{1}{2}\!\left(\!\rho\,\frac{\delta\mathcal{L}}{\delta\p}\frac{\delta\mathcal{L}}{\delta\p}\!+\!
\frac{\delta\mathcal{L}}{\delta\partial^\alpha\p}\frac{\delta\mathcal{L}}{\delta\partial_\alpha\p}\!\right)\delta^\mu_\nu b^\nu\right],
\end{split}
\end{equation}
where in the second step we used the condition \eqref{eq:ssfcondition}. Observe that $\p_\ve$  satisfies the total divergence condition \eqref{eq:totdiv} with $\psi^\mu=\tfrac{1}{2}\big[\rho\,(\delta\mathcal{L}/\delta\p)^2+(\delta\mathcal{L}/\delta\partial\p)^2\big]\,\delta^\mu_\nu b^\nu$. Hence, applying Theorem \ref{teo:NNN} we get
\begin{equation}
j^\mu_{(b)}=\left[\frac{\delta\mathcal{L}}{\delta\partial_\mu\p}
\frac{\delta\mathcal{L}}{\delta\partial^\nu\p}
-\frac{1}{2}\delta^\mu_\nu\left(\rho\,\frac{\delta\mathcal{L}}{\delta\p}\frac{\delta\mathcal{L}}{\delta\p}+
\frac{\delta\mathcal{L}}{\delta\partial^\alpha\p}\frac{\delta\mathcal{L}}{\delta\partial_\alpha\p}\right)\right] b^\nu\quad \textnormal{for all } x\in\Omega.
\end{equation}
Since there are $D$ independent generators for each possible choice of the translation parameter $b^\nu$, the final result \eqref{eq:D} is obtained by defining $j_{(b)}^\mu\equiv {K}^\mu_\nu\, b^\nu$.
\end{proof}
We leave it to the reader to verify that Theorem \ref{teo:ssf} is the covariant transposition of energy conservation for the canonical  harmonic oscillator $L=\tfrac{1}{2}|\dot{q}|^2-\tfrac{1}{2}|q|^2$ and its higher-order generalizations. 
It is also simple to prove that ${K}^{\mu\nu}=\eta^{\nu\sigma}K^\mu_\sigma$ is symmetric, namely ${K}^{\mu\nu}={K}^{\nu\mu}$.

As a final remark, we notice that condition \eqref{eq:ssfcondition} can be read as a functional PDE that restricts the form of the Lagrangian densities to which apply Theorem \ref{teo:ssf}. In this respect, Subsec. \ref{sub:spacetimeex} will use Theorem \ref{teo:ssf} with a relevant class of solutions of equation \eqref{eq:ssfcondition} whose coefficients explicitly depend on spacetime.


\subsection{\label{sec:3}Currents from the Energy-Momentum tensor}
In this section we consider theories that yield conserved currents by satisfying specific conditions on their Lagrangian density. In particular, we deduce conserved currents that are generated by the Energy-Momentum tensor. To exploit such behavior, let us introduce the following preliminary result.

\begin{theorem}\label{teo:ssf2}
Consider a Lagrangian density $\Lag$ such that $\delta\mathcal{L}/\delta x^\mu= \partial_\mu\vt $ with $\vt (x,\p,\partial\p)$.
Suppose that for a given family $(\ve,x)\mapsto \vp_\ve(x)$ with $\vp_0=0$ there exists a $D$-vector  $x\mapsto f^\alpha(x)$ such that for all motions, whether solutions or not of \eqref{eq:eom},
\begin{equation}\label{eq:cond2d}
\frac{\partial \mathcal{L}(x,\vp_\ve,\partial\vp_\ve)}{\partial\ve}\bigg\rvert_{\ve=0}=
-\left[\frac{\delta \mathcal{L}}{\delta\partial_\mu\p}\partial_\alpha\p-(\mathcal{L}-\vt)\,\delta^\mu_\alpha\right] \partial_\mu f^\alpha.
\end{equation}
Let $x\mapsto \p(x)$ be a solution of the field equation \eqref{eq:eom} for $\mathcal{L}$. Then, $j^\mu(x,\p,\partial\p)$ is a conserved current for equation \eqref{eq:eom}, namely $\partial_\mu j^{\mu}=0$, with
\begin{equation}\label{eq:d}
j^\mu\equiv
\left[\frac{\delta\mathcal{L}}{\delta\partial_\mu\p}\partial_\alpha\p-\big(\mathcal{L}-\vt\big)\,\delta^\mu_\alpha\right]  f^\alpha+\frac{\delta\mathcal{L}}{\delta\partial_\mu\p}\,\frac{\partial\vp_\ve}{\partial\ve}\bigg\rvert_{\ve=0}.
\end{equation}
\end{theorem}
\begin{proof}
Consider the mixed-family $\p_\ve=\p(x+\ve f(x))+\vp_\ve(x)$. It follows that $\partial\p_\ve/\partial\ve\rvert_{\ve=0}=\partial_\mu\p f^\mu+\partial\vp_\ve/\partial\ve\rvert_{\ve=0}$. By defining $\mathcal{L}_{\p_\ve}\equiv\Lageps$ and $\mathcal{L}_{\vp_\ve}\equiv\mathcal{L}(x,\vp_\ve,\partial\vp_\ve)$, we can compute:
\begin{equation}\label{eq:TT1}
\begin{split}
\frac{\partial \mathcal{L}_{\p_\ve}}{\partial\ve}\bigg\rvert_{\ve=0}&=\frac{\delta\mathcal{L}}{\delta\p}\frac{\partial\p_\ve}{\partial\ve}\bigg\rvert_{\ve=0}+\frac{\delta\mathcal{L}}{\delta\partial_\alpha\p}\partial_\alpha\frac{\partial\p_\ve}{\partial\ve}\bigg\rvert_{\ve=0}\\
&=\frac{\delta\mathcal{L}}{\delta\p}\left[\partial_\mu\p f^\mu+\frac{\partial\vp_\ve}{\partial\ve}\bigg\rvert_{\ve=0}\right]+\frac{\delta\mathcal{L}}{\delta\partial_\alpha\p}\left[\partial_\alpha(\partial_\mu\p f^\mu)+\partial_\alpha\frac{\partial\vp_\ve}{\partial\ve}\bigg\rvert_{\ve=0}\right]\\
&= \delta^\alpha_\mu\partial_\mu(\mathcal{L}-\vt)f^\mu+\left[\frac{\partial \mathcal{L}_{\vp_\ve}}{\partial\ve}\bigg\rvert_{\ve=0}+\frac{\delta\mathcal{L}}{\delta\partial_\alpha\p}\partial_\mu\p\,\partial_\alpha f^\mu\right]\\
& = \partial_\mu\big[\delta^\mu_\alpha(\mathcal{L}-\vt)f^\alpha\big]
\end{split}
\end{equation}
where in the second and third steps we used the definition $\delta\mathcal{L}/\delta x^\mu= \partial_\mu\vt $ and the condition \eqref{eq:cond2d} respectively. Observe that the field-shift $\p_\ve$  satisfies the total divergence condition \eqref{eq:totdiv} with the function $\psi^\mu=\delta^\mu_\alpha(\mathcal{L}-\vt)f^\alpha$. Hence, applying Theorem \ref{teo:NNN}, we get expression \eqref{eq:d}.
\end{proof}
The above Theorem \ref{teo:ssf2}  is completely general, in the sense that it does not place any assumption on the dependencies of the Lagrangian density $\Lag$. However, in the spacetime independent case $\mathcal{L}(\p,\partial\p)$ we can say more.
In particular, in that case the Energy-Momentum tensor becomes a generator of conserved currents.
\begin{corollary}\label{teo:ssf3}
Consider a spacetime-independent Lagrangian density $\mathcal{L}(\p,\partial\p)$ and suppose that for a given family $(\ve,x)\mapsto \vp_\ve(x)$ with $\vp_0=0$ there exists a $D$-vector $x\mapsto f^\alpha(x)$ such that for all motions, whether solutions or not of \eqref{eq:eom},
\begin{equation}\label{eq:cond2.1}
\frac{\partial \mathcal{L}(\vp_\ve,\partial\vp_\ve)}{\partial\ve}\bigg\rvert_{\ve=0}=-T^\mu_\alpha\,\partial_\mu f^\alpha\,,
\end{equation}
where $T^\mu_\nu$ is the Energy-Momentum tensor \eqref{eq:T} for $\mathcal{L}$.
Let $x\mapsto \p(x)$ be a solution of the field equation \eqref{eq:eom} for $\mathcal{L}$. Then, $j^\mu(x,\p,\partial\p)$ is a conserved current for equation \eqref{eq:eom}, namely $\partial_\mu j^{\mu}=0$, with
\begin{equation}\label{eq:xx}
j^\mu=T^\mu_\alpha f^\alpha+\frac{\delta\mathcal{L}}{\delta\partial_\mu\p}\frac{\partial\vp_\ve}{\partial\ve}\bigg\rvert_{\ve=0}.
\end{equation}
\end{corollary}
\begin{proof}
Consider Theorem \ref{teo:ssf2}. Being $\mathcal{L}=\mathcal{L}(\p,\partial\p)$, we have $\vt=0$. Substitute the definition of Energy-Momentum tensor \eqref{eq:T} inside expressions  \eqref{eq:cond2d} and \eqref{eq:d}.
\end{proof}


\subsection{Homogeneous potentials of degree $k=D\Delta_D^{-1}$}
\label{sec:4}

Homogeneous functions are widely used in many aspects of physics and mathematics. More specifically, if $F:\p\rightarrow F(\p)$ is a smooth function and $k$ is a number, then $F$ is said to be \textit{homogeneous of degree} $k$  in $\Omega$ if
\begin{equation}\label{eq:Uhom}
F(s\p)=s^kF(\p)\quad \textnormal{for all } s\in\mathbb{R},\,x\in\Omega.
\end{equation}
Generally, the degree of a homogeneous function can be thought of as describing how the function behaves under change of scale. 
In this respect, it is common practice to study the scaling properties of a field theory through the use of the so-called \textit{scaling dimension} $\Delta_D\equiv (D-2)/2$ of the related field \cite{CardyJ:1996}.

We noticed that in Ref. \cite{Zampieri:2013ziw} the authors proved a new approach to a well-known result by Logan \cite{Logan:1977} to find first integrals for systems $L=\tfrac{1}{2}|\dot{q}|^2-U(|q|)$ in a homogeneous potential $U(|q|)$ of degree $k=-2$. Inspired by this example, we deduce that such result can be extended to scalar field theories with $D$  spacetime dimensions. In particular, we discover that the total divergence condition seems to be too easy to satisfy for homogeneous potentials whose degree of homogeneity depends on the scaling dimension $\Delta_D$ of the field.
\begin{theorem}\label{teo:hom2}
Consider the spacetime-independent Lagrangian density
\begin{equation}\label{eq:ref}
\mathcal{L}(\p,\partial\p)=\frac{1}{2}\partial_\mu\p\,\partial^\mu\p-U(\p)\,,
\end{equation}
and assume that the potential $U(\p)$ is homogeneous of degree $k=D\Delta_D^{-1}$.
Let $x\mapsto \p(x)$ be a solution of the field equation \eqref{eq:eom} for $\mathcal{L}$, and let $T^\mu_\nu$ be the Energy-Momentum tensor \eqref{eq:T} for $\mathcal{L}$. Then, $J^\mu(x,\p,\partial\p)$ and $H^{\mu\nu}(x,\p,\partial\p)$ are conserved currents for equation \eqref{eq:eom}, namely $\partial_\mu J^{\mu}=\partial_{\mu} H^{\mu\nu}=0$, with
\begin{equation}\label{eq:homJ}
J^\mu=\Delta_D\,\p\, \partial^\mu\p+x^\alpha T^\mu_\alpha\,,
\end{equation}
\begin{equation}\label{eq:homH}
H^{\mu\nu}=\frac{1}{2}\Delta_D\,\eta^{\mu\nu}\p^2+\frac{1}{2}x^\alpha x_\alpha T^{\mu\nu}-x^\nu J^\mu\,.
\end{equation}
\end{theorem}
\begin{proof}
Consider the mixed-family $\p_\ve=e^{\,\ve\Delta_D}\p(e^{\ve}x)$. It follows that $\partial\p_\ve/\partial\ve\rvert_{\ve=0}=\Delta_D \p+x^\mu\partial_\mu\p$, so we can compute
\begin{equation}\label{eq:HD1}
\begin{split}
\frac{\partial \Lageps}{\partial\ve}\bigg\rvert_{\ve=0}&=
\frac{\delta\mathcal{L}}{\delta\p}\frac{\partial\p_\ve}{\partial\ve}\bigg\rvert_{\ve=0}+\frac{\delta\mathcal{L}}{\delta\partial_\alpha\p}\partial_\alpha\frac{\partial\p_\ve}{\partial\ve}\bigg\rvert_{\ve=0}\\
&=-U'\Big[\Delta_D \p+x^\mu\partial_\mu\p\Big]\!+\!\partial^\alpha\p\Big[(\Delta_D+1) \partial_\alpha\p+x^\mu\partial_{\alpha\mu}\p\Big]\\
&=D\left[\frac{1}{2}\,\partial_\mu\p\,\partial^\mu\p-\frac{\Delta_D}{D}\,\p \,U'\right]+x^\mu\partial_\mu \mathcal{L}\\
&=\partial_\mu\big(x^\mu\mathcal{L}\big),
\end{split}
\end{equation}
where $U'\equiv dU/d\p$. In the last step of \eqref{eq:HD1} we used the trace relation $\partial_\mu x^\mu=\delta^\mu_\mu=D$ and the \textit{Euler's Theorem} for homogeneous functions ($\p\, U'=k\,U$) assuming $k=D\Delta_D^{-1}$ by hypotesys. Observe that $\p_\ve$  satisfies the total divergence condition \eqref{eq:totdiv} with $\psi^\mu=x^\mu\mathcal{L}$. Hence, applying Theorem \ref{teo:NNN} and using the definition \eqref{eq:T} of Energy-Momentum tensor, we get our result \eqref{eq:homJ}
\begin{equation}\label{eq:HD2}
\begin{split}
J^\mu&=\partial^\mu\p\big(\Delta_D\p+x^\alpha\partial_\alpha\p\big)-\delta^\mu_\alpha x^\alpha\mathcal{L}\\
&=\Delta_D\,\p\, \partial^\mu\p+x^\alpha T^\mu_\alpha\,.
\end{split}
\end{equation}
In addition, since $\partial_\mu T^\mu_\nu=\partial_\mu J^\mu=0$, we notice that expression \eqref{eq:HD2} can be rewritten as a vanishing total divergence as follows
\begin{equation}\label{eq:HD3}
\begin{split}
0&=\Delta_D\,\p\, \partial^\mu\p+x^\alpha T^\mu_\alpha-J^\mu\\
&=\partial^\mu\big(\tfrac{1}{2}\Delta_D\p^2\big)+\partial^\nu\big(\tfrac{1}{2} x^\alpha x_\alpha T^\mu_\nu\big)-\partial_\nu\big(x^\mu J^\nu\big)\\
&=\partial_\mu\big(\tfrac{1}{2}\Delta_D\eta^{\mu\nu}\p^2+\tfrac{1}{2}x^\alpha x_\alpha T^{\mu\nu}-x^\nu J^\mu\big)\\
&\equiv\partial_\mu H^{\mu\nu}.
\end{split}
\end{equation}
Hence, also expression \eqref{eq:homH} is a conserved current.
\end{proof}

We want to emphasize that $H^{\mu\nu}$ is not symmetric. In addition, since $\Delta_2^{-1}\rightarrow \infty$, Theorem \ref{teo:hom2} holds for any dimension $D\ne 2$ of spacetime in which the theory is embedded. It should be also noticed that when $D=1$ we have $\Delta_1=-2$, hence the degree of homogeneity required by Theorem \ref{teo:hom2} to work is $k=1\cdot\Delta_1=-2$, which exactly recovers the case analyzed by Ref.  \cite{Zampieri:2013ziw}.

Interestingly, $H^{\mu\nu}$ has illuminating consequences for the dynamics. In particular, by multiplying expression \eqref{eq:homH} by $\eta_{\mu\nu}$ and solving with respect to $\p$, we get 
\begin{equation}\label{eq:DIST}
|\p(x)|=\sqrt{\frac{2 H^\mu_\mu-x^\alpha x_\alpha T^\mu_\mu-2x_\mu J^\mu}{D\Delta_D}}.
\end{equation}
This formula exactly gives how the distance from the origin $|\p(x)|$ depends on spacetime, even though we don't know the shape of the solution $x\mapsto \p(x)$ of the field equation \eqref{eq:eom}. As a consequence, the solution is born and dies at the origin when $2 H^\mu_\mu-x^\alpha x_\alpha T^\mu_\mu-2x_\mu J^\mu=0$. Again, if $D=1$, expression \eqref{eq:DIST} exactly recovers the result proposed in Ref. \cite{Zampieri:2013ziw}.


\subsection{\label{sec:6}Dissipative theories}

We noticed that in Ref. \cite{Zampieri:2013ziw} Gorni and Zampieri proved a new approach to find first integrals for ``dissipative'' systems $L=e^{ht}\big(\tfrac{1}{2}|\dot{q}|^2-U(|q|)\big)$ in a homogeneous potential $U(|q|)$ of degree $k=2$. The wording ``dissipative''  comes from the related equation of motion $\ddot{q}+h\dot{q}+U'=0$, which depends on velocities.

Inspired by this example, we deduce that such result can be extended to scalar field theories within a $D$-dimensional  spacetime.

\begin{theorem}\label{teo:diss}
Consider the Lagrangian density
\begin{equation}\label{eq:dissLag}
\Lag=e^{h_\lambda x^\lambda}\!L(\p,\partial\p)\,\qquad L=\frac{1}{2}\partial_\mu\p\,\partial^\mu\p-U(\p)\,,
\end{equation}
and assume that the potential $U(\p)$ is homogeneous of degree $k=2$. Let $x\mapsto \p(x)$ be a solution of the field equation \eqref{eq:eom} for $\mathcal{L}$, and let $T^{\mu\nu}_L$ be the (not conserved) Energy-Momentum tensor \eqref{eq:T} for $L$. Then, $j^\mu(x,\p,\partial\p)$ is a conserved current for equation \eqref{eq:eom}, namely $\partial_\mu j^{\mu}=0$, with
\begin{equation}\label{eq:jjj}
j^\mu=e^{h_\lambda x^\lambda}\!\left(\p\,\partial^\mu\p+2h^{-1}_\nu T_L^{\mu\nu}\right).
\end{equation}
\end{theorem}
\begin{proof}
Consider the mixed-family $\p_\ve=\p+\ve(\p\!+\!c^\mu\partial_\mu\p)$ with the constant parameter $c^\mu\in\mathbb{R}^D$. It follows that $\partial\p_\ve/\partial\ve\rvert_{\ve=0}=\p+c^\mu\partial_\mu\p$, so we can compute
\begin{equation}\label{eq:xxx}
\begin{split}
\frac{\partial \Lageps}{\partial\ve}\bigg\rvert_{\ve=0}&=
\frac{\delta\mathcal{L}}{\delta\p}\frac{\partial\p_\ve}{\partial\ve}\bigg\rvert_{\ve=0}+\frac{\delta\mathcal{L}}{\delta\partial_\mu\p}\partial_\mu\frac{\partial\p_\ve}{\partial\ve}\bigg\rvert_{\ve=0}\\
&=e^{h_\lambda x^\lambda}\!\left[-U'\big(\p+c^\mu\partial_\mu\p\big)+\partial^\mu\p\big(\partial_\mu\p+c^\alpha\partial_{\alpha\mu}\p\big)\right]\\
&=e^{h_\lambda x^\lambda}\!\left[2\left(\frac{1}{2}\partial_\mu\p\partial^\mu\p-\frac{1}{2}\p\, U'\right)+c^\mu\partial_\mu L\right]\\
&=c^\mu\!\left[\big(2c_\mu^{-1}e^{h_\lambda x^\lambda}\big)L+e^{h_\lambda x^\lambda}\partial_\mu L\right]\\
&=\partial_\mu\left(2\eta^{\mu\nu}h_\nu^{-1}e^{h_\lambda x^\lambda} L\right),
\end{split}
\end{equation}
where $U'\equiv dU/d\p$. In the second-last step of \eqref{eq:xxx} we used the relation $c^\mu c_\mu^{-1}=1$ and the \textit{Euler's Theorem} for homogeneous functions ($\p\, U'=k\,U$) assuming $k=2$ by hypothesis. In the last step of \eqref{eq:xxx} we have forced $c^\mu$ to satisfy the relation $2c_\mu^{-1}\equiv h_\mu$, hence $c_\mu=2h_\mu^{-1}$. Observe that $\p_\ve$  satisfies the total divergence condition \eqref{eq:totdiv} with $\psi^\mu=2\eta^{\mu\nu}h_\nu^{-1}e^{h_\lambda x^\lambda} L$. Hence, applying Theorem \ref{teo:NNN} and using the definition \eqref{eq:T} of Energy-Momentum tensor with respect to $L$, we get 
\begin{equation}
\begin{split}
j^\mu&=e^{h_\lambda x^\lambda}\partial^\mu\p\big(\p+2h_\nu^{-1}\partial^\nu\p\big)-2\eta^{\mu\nu}h_\nu^{-1}e^{h_\lambda x^\lambda} L\\
&=e^{h_\lambda x^\lambda}\left[\p\,\partial^\mu\p+2h_\nu^{-1}\left(\partial^\mu\p\partial^\nu\p-\eta^{\mu\nu}L\right)\right]\\
&=e^{h_\lambda x^\lambda}\left(\p\,\partial^\mu\p+2h_\nu^{-1}T^{\mu\nu}_L\right).
\end{split}
\end{equation}
Hence, expression \eqref{eq:jjj} is a conserved current.
\end{proof}
We leave it to the reader to verify the first integral obtained in Ref. \cite{Zampieri:2013ziw} for dissipative systems is a trivial consequence of Theorem \ref{teo:diss}, with $q\in\mathbb{R}^m$ and $D=1$. Our result \eqref{eq:jjj} provides a powerful perspective to get conserved currents in scalar field theories, as we will show in Subsec. \ref{sec:dissKG}.


\subsection{\label{sec:5}Theories with finite invariances}

Another interesting situation generating conserved currents arises when the Lagrangian density, once evaluated on a perturbed field $\p_\ve$, has constant derivative at $\ve = 0$. Following the nomenclature of Gorni and Zampieri \cite{Zampieri:2013ziw}, this condition will be called \textit{finite invariance}.
\begin{theorem} \label{teo:J}
Let $x\mapsto \p(x)$ be a solution of the field equation \eqref{eq:eom} and suppose that for a given perturbed field $\p_\ve$ there exists a constant parameter  $\xi\in \mathbb{R}$ such that
\begin{equation}\label{eq:lxcr}
\frac{\partial \Lageps}{\partial\ve}\bigg\rvert_{\ve=0}=\xi\,.
\end{equation}
Then, $j^\mu(x,\p,\partial\p)$ is a conserved current for equation \eqref{eq:eom}, namely $\partial_\mu j^{\mu}=0$, with
\begin{equation}\label{eq:lxc}
j^\mu=\frac{\delta \mathcal{L}}{\delta\,\partial_\mu\p}\,\frac{\partial\p_\ve}{\partial\ve}\bigg\rvert_{\ve=0}-\xi x^\mu\,.
\end{equation}
\end{theorem}
\begin{proof}
Observe from \eqref{eq:lxcr} that our perturbed field $\p_\ve$  satisfies the total derivative condition \eqref{eq:totdiv} with $\psi^\mu=\xi x^\mu$. Hence, applying Theorem \ref{teo:NNN}, we get expression \eqref{eq:lxc}.
\end{proof}
Standard working  cases for Theorem \ref{teo:J} are Lagrangian densities such that $\Lageps$ does not depend on $\ve$, i.e. are invariant under the perturbed field considered.


\section{\label{sec:app}Applications}

Nonlocal constants naturally arise in a number of contexts.
In the present section we deal with some neat applications of our theorems to well-known scalar field theories. In particular,  standard Noetherian results are recovered by our nonlocal machinery.


\subsection{\label{sub:KG}Nonlinear interacting (real) field theory}

We start to examine scalar field theories that include self-interaction terms in a 4-dimensional ($D=4$) spacetime. Restricting ourselves to interaction terms containing polynomial functions of a real scalar field $\p$ only (and not its derivatives), let us consider the following Lagrangian density
\begin{equation}\label{eq:nl}
\mathcal{L}_{(m,g)}=\frac{1}{2}\partial_\mu\p\,\partial^\mu\p-U_{(m,g)}\,,
\end{equation}
with
\begin{equation}\label{eq:nl2}
U_{(m,g)}=\frac{1}{2} m^2\p^2+\sum_{n=3}^\infty\frac{1}{n!}\,g_n\,\p^n\,.
\end{equation}
Here, $m$ is the \textit{mass parameter} of the field $\p$, whereas the \textit{coupling constants} $\{g_n\}_{n=3...\infty}$ are a measure of the strength of the field interactions. Among all, special cases of theory \eqref{eq:nl} are: (i) the \textit{free theory} $\mathcal{L}_{(0,0)}$, which forms the basis for doing perturbative calculations \cite{Faber:2001dv}; (ii) the free \textit{Klein-Gordon theory} $\mathcal{L}_{(m,0)}$, which describes the behavior of free spinless quantum fields \cite{Chiarelli:2017rdc}; and (iii) the $\phi^4$-\textit{theory} $\mathcal{L}_{(m,g_4)}$, that illustrates a number of interesting phenomena like the self-interaction of the Higgs Field in the standard electroweak theory \cite{Brandt:1975xn} or the statistical mechanics of nonlinear coherent structures such as solitons and solitary waves \cite{osti_6383942}.

The field equation \eqref{eq:eom} for the Lagrangian density \eqref{eq:nl} reads
\begin{equation}\label{eq:ninEOM}
\Box\p+m^2\p+\sum_{n=3}^\infty \frac{1}{(n-1)!}\,g_n\,\p^{n-1}=0\,,
\end{equation}
where $\Box\equiv \partial^\mu\partial_\mu$ is the d'Alembert operator. We want to emphasize that equation \eqref{eq:ninEOM} is linear if $\{g_n\}_{n=3...\infty}=0$, which means that within the free Klein-Gordon theory $\mathcal{L}_{(m,0)}$ quantum particles exist in their isolated modes.

Let us try some random perturbed fields for Theorem \ref{teo:nonloconst}. The first family is $\p_\ve=\p-\ve$. We compute $\partial \p_\ve/\partial\ve |_{\ve=0}=-1$, so the related nonlocal constant \eqref{eq:nonloconst} is
\begin{equation}\label{eq:asa}
-\oint_{\partial\Omega}\!d^{3}x\, n_\mu\partial^\mu\p
-\int_{\Omega}\!d^4\!x\left[m^2\p+\sum_{n=3}^\infty \frac{1}{(n-1)!}\,g_n\,\p^{n-1}\right]\,.
\end{equation}
Clearly, $\p$ can not be a total divergence of some $D$-vector $\psi^\mu$ being the same for all smooth paths. However, since $\p$ is a solution of the Euler-Lagrange equation \eqref{eq:eom}, equation \eqref{eq:ninEOM} leads to rewrite expression \eqref{eq:asa} as
\begin{equation}\label{eq:asd}
-\oint_{\partial\Omega}\!d^{3}x\, n_\mu\partial^\mu\p
+\int_{\Omega}\!d^4\!x\,\partial_\mu\big(\partial^\mu\p\big)\,.
\end{equation}
that, applying the covariant Gauss divergence theorem to the second term, returns a nonlocal constant which is trivially 0.

Let us search for a second nonlocal constant starting from the
field-shift family $\p_\ve=(1+\ve)\,\p$, that gives $\partial \p_\ve/\partial\ve |_{\ve=0}=\partial_\mu\p$. Using Theorem \ref{teo:nonloconst} we compute
\begin{equation}\label{eq:asasdf}
\oint_{\partial\Omega}\!d^{3}x\, n_\mu\p\,\partial^\mu\p
-\int_{\Omega}\!d^4\!x\left\{\partial^\mu\p\,\partial_\mu\p-\p\left[m^2\p+\sum_{n=3}^\infty \frac{1}{(n-1)!}\,g_n\,\p^{n-1}\right]\right\}.
\end{equation}
Even in this case, the integrand never seems satisfying the total derivative condition \eqref{eq:totdiv}. However, since $\p$ is a solution of the Euler-Lagrange equation \eqref{eq:eom}, equation \eqref{eq:ninEOM} leads to rewrite the second integrand of \eqref{eq:asasdf} as $\partial^\mu\p\,\partial_\mu\p+\p\Box\p$, which simplifies expression \eqref{eq:asasdf} to
\begin{equation}\label{eq:asd234}
\oint_{\partial\Omega}\!d^{3}x\, n_\mu\p\,\partial^\mu\p
-\int_{\Omega}\!d^4\!x\,\partial_\mu\big(\p\,\partial^\mu\p\big)\,.
\end{equation}
Applying the covariant Gauss divergence theorem to the second term, the resulting nonlocal constant is again trivially 0.

Hence, with respect to the interacting theory considered, what matters is that Theorem \ref{teo:nonloconst} is not powerful enough to yield non trivial nonlocal constants and related conserved currents. However, we come to Theorem 3.3, which is the most natural to consider as our Lagrangian \eqref{eq:nl} is explicitly spacetime-independent. Indeed, after a couple of calculations, we obtain from expression \eqref{eq:T} our first non trivial conserved current in terms of the well-known Energy-Momentum tensor \cite{Forger:2003ut}:
\begin{equation}\label{eq:tT}
T^{\mu}_{\nu\,(m,g)}=\partial^\mu\p\,\partial_\nu\p-\delta^\mu_\nu\left[\frac{1}{2}\partial_\alpha\p\,\partial^\alpha\p-\frac{1}{2} m^2\p^2-\sum_{n=3}^\infty\frac{1}{n!}\,g_n\,\p^n\right].
\end{equation}

Let's continue with our search. It is easy to prove that our Lagrangian density \eqref{eq:nl}, if restricted to the free Klein-Gordon case $\mathcal{L}_{(m,0)}$, satisfies the hypothesis \eqref{eq:ssfcondition} of validity of the Theorem \ref{teo:ssf} with $\rho=1/(2m^{2})$. Hence, we get a conserved current from expression \eqref{eq:D}, that we discover to have  the same expression of the  Energy-Momentum tensor \eqref{eq:tT}
\begin{equation}
{K}^\mu_{\nu\,(m,0)}\!=T^{\mu}_{\nu\,(m,0)}.
\end{equation}
In Subsec. \ref{sub:spacetimeex} we will analyze a simple model where such degeneracy can be removed.

An additional conserved current can be recovered by Corollary \ref{teo:ssf3} as follows. Let us consider $\vp_\ve\equiv\ve\p$ and $f^\alpha(x)\equiv x^\alpha$. Computing separately both the sides of the mathematical equality \eqref{eq:cond2.1} we get
\begin{equation}
\label{eq:x11}
\begin{cases}
\displaystyle
\frac{\partial \mathcal{L}(\vp_\ve,\partial\vp_\ve)}{\partial\ve}\bigg\rvert_{\ve=0}= \partial_\mu\p\,\partial^\mu\p- m^2\p^2-\sum_{n=3}^\infty\frac{1}{(n-1)!}\,g_n\,\p^n\\
\displaystyle
T^\mu_{\alpha\,(m,g)}\,\partial_\mu f^\alpha= -\Delta_4\partial_\mu\p\,\partial^\mu\p+ 2m^2\p^2+\sum_{n=3}^\infty\frac{4}{n!}\,g_n\,\p^n
\end{cases},
\end{equation}
with $\Delta_4\equiv(4-2)/2=1$ the scaling dimension defined in Sec. \ref{sec:4}.
By comparing the numerical coefficients, it is simple to see that condition \eqref{eq:cond2.1} is satisfied by  expression \eqref{eq:x11} if $m=0$ and $n=D=4$, namely by the $4$-dimensional massless $\phi^4$-theory $\mathcal{L}_{(0,g_4)}$. 
Hence, computing the Energy-Momentum tensor \eqref{eq:tT} inside equation \eqref{eq:xx}, we get 
\begin{equation}\label{eq:jv}
j^\mu_{(0,g_4)}= \partial^\mu\p\, \big(x^\alpha\partial_\alpha\p+\p\big)-x^\mu\left(\frac{1}{2}\partial_\mu\p\,\partial^\mu\p-\frac{1}{4!}\,g_4\,\p^4\right)\,.
\end{equation}

Being theory \eqref{eq:nl} embedded in a 4-dimensional spacetime, a homogeneous potential \eqref{eq:nl2} of degree $k=4\cdot \Delta_4^{-1}=4$
makes our Theorem \ref{teo:hom2} working to generate two conserved currents. This is the case for $U_{(0,g_4)}$, corresponding again to the massless $\phi^4$-theory $\mathcal{L}_{(0,g_4)}$, which yields
\begin{equation}\label{eq:jkly}
J^\mu_{(0,g_4)}=j^\mu_{(0,g_4)},
\end{equation}
and
\begin{equation}\label{eq:homHfdfd}
H_{(0,g_4)}^{\mu\nu}=\frac{1}{2}\eta^{\mu\nu}\p^2+\frac{1}{2}x^\alpha x_\alpha T_{(0,g_4)}^{\mu\nu}-x^\nu J_{(0,g_4)}^\mu\,.
\end{equation}
With reference to the massless $\phi^4$-theory, the dependence of $|\p(x)|$ on spacetime is provided by equation \eqref{eq:DIST}, which in our case yields
\begin{equation}\label{eq:DISTmass}
|\p_{(0,g_4)}|=\frac{1}{2}\sqrt{2 H^\mu_{\mu\,(0,g_4)}-x^\alpha x_\alpha T^\mu_{\mu\,(0,g_4)}-2x_\mu J_{(0,g_4)}^\mu}\,.
\end{equation}

We leave it to the reader to prove that the total divergence of the three  conserved currents \eqref{eq:tT}, \eqref{eq:jv}, and \eqref{eq:homHfdfd} is identically zero.


\subsection{\label{sub:CKG}Nonlinear interacting (complex) field theory}

Complex scalar fields generally represent spin-0 particles and antiparticles with electric charge or  other properties (such as strangeness, charm, etc.).
A nonlinear interacting (complex) field theory can be easily derived by a variational method considering the following Lagrangian density
\begin{equation}\label{eq:nlC}
\mathcal{L}_{(m,v)}=\partial_\mu\p\,\partial^\mu\pc-U_{(m,v)}\,.
\end{equation}
Here,
\begin{equation}\label{eq:nl2C}
U_{(m,v)}=m^2|\p|^2+\sum_{n=2}^\infty\frac{1}{n!}\,v_{2n}\,|\p|^{2n}\,,
\end{equation}
where $\p\in\mathbb{C}$ is a complex scalar field, $m\in\mathbb{R}$ is its mass, and $\{v_{2n}\}_{n=2...\infty}\in\mathbb{R}$ are coupling constants. 
In our notation $\pc$ is the complex-conjugate of $\p$, whereas $|\p|^2\equiv \p \pc\in \mathbb{R}$ is the square modulus.

By substitution of theory \eqref{eq:nlC} in equation \eqref{eq:eom}, one has the following Euler-Lagrange equations with respect to $\p$ and $\pc$
\begin{equation}
\label{eq:EOMC1}
\begin{cases}
\displaystyle
\Box\pc+m^2\pc+\sum_{n=2}^\infty\frac{2}{(n-1)!}\,v_{2n}\,|\p|^{2n-1}\pc=0\\
\displaystyle
\Box\p+m^2\p+\sum_{n=2}^\infty\frac{2}{(n-1)!}\,v_{2n}\,|\p|^{2n-1}\p=0
\end{cases},
\end{equation}
where $\Box\equiv \partial^\mu\partial_\mu$ is the d'Alembert operator.

As in the real case, the Lagrangian density \eqref{eq:nlC} is clearly spacetime-independent. 
Hence, the conserved current generated by a space-time translation in Theorem \ref{teo:trasl} becomes the Energy-Momentum tensor
\begin{equation}\label{eq:tTC}
T^{\mu}_{\nu\,(m,v)}=\partial^\mu\pc\,\partial_\nu\p+\partial^\mu\p\,\partial_\nu\pc-\delta^\mu_\nu\left[\partial_\alpha\p\,\partial^\alpha\pc-m^2|\p|^2+\sum_{n=2}^\infty\frac{1}{n!}\,v_{2n}\,|\p|^{2n}\right].
\end{equation}

In the same way as the above Subsec. \ref{sub:KG}, Theorem \ref{teo:ssf} holds for theory \eqref{eq:nlC} with the constant parameter $\rho=-m^{-2}$, which yields ${K}^\mu_{\nu\,(m,0)}\!=T^{\mu}_{\nu\,(m,0)}$.
Again, this degeneracy should not surprise as we will discuss in Subsec. \ref{sub:spacetimeex}.

We point out that also Theorem \ref{teo:hom2} turns out extremely useful if we select from \eqref{eq:nlC} the $n=2$ (massless) theory, which exhibits a homogeneous potential $U_{(0,v_4)}$ of degree $k=4\cdot \Delta_4^{-1}=4$. In fact, direct computations show that there exist two conserved currents
\begin{equation}\label{eq:jkl}
J^\mu_{(0,v_4)}=\p\,\partial^\mu\pc\!+\!\pc\partial^\mu\p\!+\!x^\alpha\Big[\partial^\mu\pc\partial_\alpha\p+\partial^\mu\p\partial_\alpha\pc\!+\!\delta_\mu^\alpha\big(\partial_\nu\p\partial^\nu\pc\!-\!\tfrac{1}{2}v_4|\p|^4\big)\Big],
\end{equation}
and
\begin{equation}\label{eq:homHfdfdC}
H_{(0,v_4)}^{\mu\nu}=\tfrac{1}{2}\eta^{\mu\nu}|\p|^2+\tfrac{1}{2}x^\alpha x_\alpha T_{(0,v_4)}^{\mu\nu}-x^\nu J_{(0,v_4)}^\mu\,.
\end{equation}

With reference to the $n=2$ theory, the dependence of $|\p(x)|$ on spacetime is provided by equation \eqref{eq:DIST}, which in our case yields
\begin{equation}\label{eq:DISTmassc}
|\p_{(0,v_4)}|=\frac{1}{2}\sqrt{2 H^\mu_{\mu\,(0,v_4)}-x^\alpha x_\alpha T^\mu_{\mu\,(0,v_4)}-2x_\mu J_{(0,v_4)}^\mu}\,.
\end{equation}

Finally, we consider the unitary family $\p_\ve=e^{i\ve}\p$ and $\pc_\ve=e^{-i\ve}\pc$, where $i$ is the imaginary unit. It is clear that, when evaluated on $\p_\ve$ and $\pc_\ve$, the Lagrangian density \eqref{eq:nlC}
does not depend on $\ve$. Since  $\partial \p_\ve/\partial\ve |_{\ve=0}=i\p$ and $\partial \pc_\ve/\partial\ve |_{\ve=0}=-i\pc$, our Theorem \ref{teo:J} gives the conserved current
\begin{equation}\label{eq:aw}
j^\mu=i\left(\,\p\,\partial^\mu\pc-\partial^\mu\p\,\pc\right).
\end{equation}
In quantum field  theory, such current is related to the number of particles created and destroyed by the creation and annihilation operators in a given state.

It is left as an exercise to show that the total divergence of the four conserved currents \eqref{eq:tTC}, \eqref{eq:jkl}, \eqref{eq:homHfdfdC}, and \eqref{eq:aw} is identically zero.


\subsection{\label{sec:dissKG}Dissipative (real) Klein-Gordon theory}

Inspired by formula \eqref{eq:dissLag}, let us introduce the  dissipative (real) Klein-Gordon Lagrangian density as follows
\begin{equation}
\label{eq:disskgrd}
\Lag=e^{h_\lambda x^\lambda}\!\left(\tfrac{1}{2}\partial_\mu\p\,\partial^\mu\p-\tfrac{1}{2}m^2\p^2\right),
\end{equation}
with $m\in\mathbb{R}$ and $h_\mu\in\mathbb{R}^D$ two real constant parameters.
The Euler-Lagrange equation \eqref{eq:eom} for expression \eqref{eq:disskgrd} is
\begin{equation}
\Box\p+h_\mu\partial^\mu \p+m^2\p=0\,,
\end{equation}
where we see a viscous resistance term $h_\mu\partial^\mu \p$, which justifies our use of the term ``dissipative''. Since the potential $U(\p)\equiv\tfrac{1}{2}m^2\p^2$ is homogeneous of degree $k=2$, Theorem \ref{teo:diss} tells us that we can obtain a conserved current  using expression \eqref{eq:jjj}, that is
\begin{equation}\label{eq:dissjs}
j^\mu=e^{h_\lambda x^\lambda}\!\Big\{\p\,\partial^\mu\p+2h^{-1}_\nu \!\left[\partial^\mu\p\,\partial_\nu\p-\delta^\mu_\nu\left(\tfrac{1}{2}\partial_\alpha\p\,\partial^\alpha\p-\tfrac{1}{2} m^2\p^2\right)\right]\Big\}.
\end{equation}


\subsection{Spacetime-dependent field theories}
\label{sub:spacetimeex}

In the previous Subsecs. \ref{sub:KG} and \ref{sub:CKG}, both Theorem \ref{teo:trasl} and Theorem \ref{teo:ssf} led to the same expression for ${K}^\mu_{\nu}$ and $T^{\mu}_{\nu}$ when applied to the Klein-Gordon theory $\mathcal{L}_{(m,0)}$.  
Investigating such behavior, one quickly observes that the degeneracy between ${K}^\mu_{\nu}$ and $T^{\mu}_{\nu}$ is a direct consequence of considering a canonical spacetime-independent theory $\mathcal{L}(\p,\partial\p)= \tfrac{1}{2}\partial_\mu\partial^\mu\p-U(\p)$ that, like $\mathcal{L}_{(m,0)}$, satisfies the condition \eqref{eq:ssfcondition}.

On the other hand, condition \eqref{eq:ssfcondition} can be read as functional Partial Differential Equation which restricts the form of the Lagrangian densities to which apply Theorem \ref{teo:ssf}. 
This fact, in contrast to what required by Theorem \ref{teo:trasl} to work, is independent of considering a spacetime-independent theory. As a consequence, the degeneracy between ${K}^\mu_{\nu}$ and $T^{\mu}_{\nu}$ can be removed considering non-canonical Lagrangian densities whose terms have an explicit dependence on spacetime. In such case, only Theorem \ref{teo:ssf} can be applied since the theory is not spacetime-independent. In this respect, a relevant non-canonical class of solutions satisfying the condition \eqref{eq:ssfcondition} is
\begin{equation}
\label{eq:Lspacetilev}
\mathcal{L}=\frac{1}{2}\,\partial^\mu\p\,\partial_\mu\p+\partial^\mu a(x)\partial_\mu\p-\frac{1}{2}\sigma\p^2-\sigma a(x)\p+c\,,
\end{equation}
with $x\mapsto a(x)$ an arbitrary smooth function, and $\sigma,c\in\mathbb{R}$ free constant parameters.
The field equation \eqref{eq:eom} for  theory \eqref{eq:Lspacetilev} reads
\begin{equation}
\label{eq:eomLspacetime}
\Box A(x) +\sigma A(x)=0\,,
\end{equation}
where $A(x)\equiv\p+a(x)$ and, as usual, $\Box\equiv\partial^\mu\partial_\mu$.
It is  easy to prove that our Lagrangian density \eqref{eq:Lspacetilev} satisfies the hypothesis \eqref{eq:ssfcondition} of the Theorem \ref{teo:ssf} with $\rho=-\sigma^{-1}$. Hence, only
\begin{equation}
2\,{K}^\mu_\nu=2\,\partial^\mu A(x)\,\partial_\nu A(x)+\delta^\mu_\nu\Big[\sigma A^2(x)-\partial^\alpha A(x)\,\partial_\alpha A(x)\Big],
\end{equation}
is conserved.

\section*{Acknowledgments}
The author would like to thank Professor Gaetano Zampieri for useful discussions.

\nocite{*}
\providecommand{\bysame}{\leavevmode\hbox to3em{\hrulefill}\thinspace}


\end{document}